\documentclass[twocolumn, amsthm]{autart}

\usepackage{graphicx}
\usepackage{amsmath}
\usepackage{amssymb,mathtools}
\usepackage{amsthm}
\usepackage{subcaption}
\usepackage{wrapfig}
\usepackage[bitstream-charter]{mathdesign}
\usepackage{hyperref}
\usepackage{stmaryrd}
\hypersetup{colorlinks=true,
citecolor=blue,linkcolor=blue
}
\usepackage{natbib}
\newtheorem{theorem}{Theorem}
\newtheorem{lemma}{Lemma}
\newtheorem{remark}{Remark}
\newtheorem{example}{Example}
\newtheorem{proposition}{Proposition}
\newtheorem{corollary}{Corollary}

\begin{document}

\begin{frontmatter}
\title{Virtual Nonholonomic Constraints: A Geometric Approach}

\author[AA]{Alexandre Anahory Simoes},  \ead{alexandre.anahory@car.upm-csic.es}
\author[LC]{Efstratios Stratoglou}, \ead{ef.stratoglou@alumnos.upm.es} \author[AB]{Anthony Bloch}, \ead{abloch@umich.edu} 
\author[AA]{Leonardo J. Colombo} \ead{leonardo.colombo@car.upm-csic.es}

\address[AA]{Centre for Automation and Robotics (CSIC-UPM), Ctra. M300 Campo Real, Km 0,200, Arganda
del Rey - 28500 Madrid, Spain.}
\address[AB] {Department of Mathematics, University of Michigan, Ann Arbor, MI 48109, USA.}
\address[LC]{Universidad Polit\'ecnica de Madrid (UPM), José Gutiérrez Abascal, 2, 28006 Madrid, Spain.}


\begin{abstract}
Virtual constraints are invariant relations imposed on a control system via feedback as opposed to real physical constraints acting on the system.  Nonholonomic systems are mechanical systems with non-integrable constraints on the velocities. In this work, we introduce the notion of \textit{virtual nonholonomic constraints} in a geometric framework. More precisely, it is a controlled invariant distribution associated with an affine connection mechanical control system. We demonstrate the existence and uniqueness of a control law defining a virtual nonholonomic constraint and we characterize the trajectories of the closed-loop system as solutions of a mechanical system associated with an induced constrained connection. Moreover, we characterize the dynamics for nonholonomic systems in terms of virtual nonholonomic constraints, i.e., we characterize when can we obtain nonholonomic dynamics from virtual nonholonomic constraints.
\end{abstract} 

\end{frontmatter}

\section{Introduction}
Virtual constraints are relations on the configuration variables of a control system which are imposed through feedback control and the action of actuators, instead of through physical connections such as gears or contact conditions with the environment. The class of virtual holonomic constraints became popular in applications to biped locomotion where it was used to express a desired walking gait (see for instance  \cite{chevallereau2009asymptotically,la2013stable,razavi2016symmetric,chevallereau2018self}), as well as for motion planning to search for periodic orbits and its employment in the technique of transverse linearization to stabilize such orbits \cite{freidovich2008periodic,westerberg2009motion,shiriaev2010transverse,mohammadi2018dynamic,nielsen2008local,consolini2010path,consolini2013control}.

Virtual nonholonomic constraints are a class of virtual constraints that depend on velocities rather than only on the configurations of the system. Such constraints were introduced in \cite{griffin2015nonholonomic}, \cite{griffin2017nonholonomic} to design a velocity-based swing foot placement in bipedal robots. In particular, this classes of virtual constraints has been used in \cite{horn2018hybrid,hamed2019nonholonomic,horn2020nonholonomic,horn2021nonholonomic} to encode velocity-dependent stable walking gaits via momenta conjugate to the unactuated degrees of freedom of legged robots and prosthetic legs.  

From a theoretical perspective, virtual constraints extend the application of zero dynamics to feedback design (see for instance \cite{isidori} and \cite{westervelt2018feedback}). In particular, the class of virtual holonomic constraints applied to mechanical systems has built rich theoretical foundations and applications in the last decade (see \cite{maggiore2012virtual,mohammadi2013lagrangian,mohammadi2015maneuvering,mohammadi2017lagrangian,mohammadi2018dynamic,vcelikovsky2015flatness,vcelikovsky2016collocated,vcelikovsky2017collocated,consolini2015induced,consolini2018coordinate}), nevertheless there is a lack of a rigorous definition and qualitative description for the class of virtual nonholonomic constraints in contrast with the holonomic situation. The recent work \cite{moran2021energy} shows a first approach to define rigorously virtual nonholonomic constraints, but the nonlinear nature of the constraints makes difficult a thorough mathematical analysis. In this work, we provide a formal definition of linear virtual nonholonomic constraints, i.e., constraints that are linear on the velocities. This particular case includes most of the examples of nonholonomic constraints in the literature of nonholonomic systems (see \cite{bloch2003nonholonomic} and \cite{ne_mark2004dynamics} for instance). Our definition is based on the invariance property under the closed-loop system and coincides with the one of \cite{moran2021energy}, in the linear case. 

In particular, a virtual nonholonomic constraint is described by a non-integrable distribution on the configuration manifold of the system for which there is a feedback control making it invariant under the flow of the closed-loop system. We provide sufficient conditions for the existence and uniqueness of such a feedback law defining the virtual nonholonomic constraint and we also characterize the trajectories of the closed-loop system as solutions of a mechanical system associated with an induced constrained connection. Moreover, we are able to produce nonholonomic dynamics by imposing virtual nonholonomic constraints on a mechanical control system. This last result allows one to control the system to satisfy desired stability properties that are well known in the literature on nonholonomic systems, through the imposition of suitable virtual nonholonomic constraints.

The remainder of the paper is structured as follows. Section \ref{sec2} introduces nonholonomic systems. We define virtual nonholonomic constraints in Section \ref{sec:controler}, where we provide sufficient conditions for the existence and uniqueness of a control law defining a virtual nonholonomic constraint, and provide examples and comparisons with the literature.  In Section \ref{sec4}, we introduce a constrained connection to characterize the closed-loop dynamics as a solution of the mechanical system associated with such a constrained connection. In Section \ref{sec5}, we show that if the input distribution is orthogonal to the virtual nonholonomic constraint distribution then the constrained dynamics is precisely the nonholonomic dynamics with respect to the original Lagrangian function. Conclusions are given in Section \ref{sec:conclusions}.

\section{Nonholonomic mechanical systems}\label{sec2}
Let $Q$ be the configuration space of a mechanical system, a differentiable manifold with $\dim(Q)=n$, and with local coordinates denoted by $(q^i)$ for $i=1,\ldots,n$. Most nonholonomic systems have linear constraints on velocities, and these are the ones we will consider. Linear constraints on the
velocities (or Pfaffian constraints) are locally given by equations of the form \begin{equation}\label{NHconstraint}\phi(q^i, \dot{q}^i)=\mu_i(q)\dot{q}^i=0,\end{equation} depending in general, on the configurations and
velocities of the system (see \cite{bloch2003nonholonomic} for instance). 

From a geometric point of view, these constraints are defined by a regular distribution ${\mathcal D}$ on
$Q$ of constant rank $(n-m)$ such that the annihilator of ${\mathcal
D}$, denoted by $\mathcal {D}^o$, is locally given at each point of $Q$ by
${\mathcal D}^o_{q} = \operatorname{span}\left\{ \mu^{a}(q)=\mu_i^{a}dq^i \; ; 1 \leq a
\leq m \right\}$, where $\mu^{a}$ are linearly independent differential one-forms at each point of $Q$. We further denote by $\Omega^{1}(Q)$ the set of differential one-forms on $Q$.

Next, consider  mechanical systems where the Lagrangian is of mechanical type, that is, mechanical systems with a dynamics described by a Lagrangian function $L:TQ\rightarrow\mathbb{R}$ which is defined by
\begin{equation}\label{mechanical:lagrangian}
    L(v_q)=\frac{1}{2}\mathcal{G}(v_q, v_q) - V(q),
\end{equation}
with $v_q\in T_qQ$, where $\mathcal{G}$ denotes a Riemannian metric on $Q$ representing the kinetic energy of the systems, $ T_qQ$, the tangent space at the point $q$ of $Q$, and
$V:Q\rightarrow\mathbb{R}$ is a (smooth) potential function, and also assume the Lagrangian system is subject to the nonholonomic constraints  given by \eqref{NHconstraint}.

\begin{defn}\label{nonholonomicsystem}
A \textit{nonholonomic mechanical system} on a smooth manifold $Q$ is given
by the triple $(\mathcal{G}, V, \mathcal{D})$, where $\mathcal{G}$ is
a Riemannian metric on $Q,$ representing the kinetic energy of the
system, $V:Q\rightarrow\mathbb{R}$ is a smooth function representing the potential
energy, and $\mathcal{D}$ a regular distribution on $Q$
describing the nonholonomic constraints.
\end{defn}

Denote by $\tau_{\mathcal{D}}:\mathcal{D}\rightarrow Q$ the canonical
projection from $\mathcal{D}$ to $Q$, locally given by $\tau_{\mathcal{D}}(q^i, \dot{q}^i)=q^i$, and denote by 
$\Gamma(\tau_{\mathcal{D}})$ the set of sections of $\tau_{D}$, that is, $Z\in\Gamma(\tau_{\mathcal{D}})$ if $Z:Q\to\mathcal{D}$ satisfies $(\tau_{\mathcal{D}}\circ Z)(q)=q$. We also denote by $\mathfrak{X}(Q)$ the set of vector fields on $Q$. If $X, Y\in\mathfrak{X}(Q),$ then
$[X,Y]$ denotes the standard Lie bracket of vector fields.

In any Riemannian manifold, there is a unique connection $\nabla^{\mathcal{G}}:\mathfrak{X}(Q)\times \mathfrak{X}(Q) \rightarrow \mathfrak{X}(Q)$ called the \textit{Levi-Civita connection} satisfying the following two properties:
\begin{enumerate}
\item $[ X,Y]=\nabla_{X}^{\mathcal{G}}Y-\nabla_{Y}^{\mathcal{G}}X$ (symmetry)
\item $X(\mathcal{G}(Y,Z))=\mathcal{G}(\nabla_{X}^{\mathcal{G}}(Y,Z))+\mathcal{G}(Y,\nabla_{X}^{\mathcal{G}}Z)$ (compatibility of the metric).
\end{enumerate}
The trajectories $q:I\rightarrow Q$ of a mechanical Lagrangian determined by a Lagrangian function as in \eqref{mechanical:lagrangian} satisfy the following equation
\begin{equation}\label{ELeq}
    \nabla_{\dot{q}}^{\mathcal{G}}\dot{q} + \text{grad}_{\mathcal{G}}V(q(t)) = 0.
\end{equation}
Observe that if the potential function vanishes, then the trajectories of the mechanical system are just the geodesics with respect to the connection $\nabla^{\mathcal{G}}$. Here, the vector field $\text{grad}_{\mathcal{G}}V\in\mathfrak{X}(Q)$ is characterized by $$\mathcal{G}(\text{grad}_{\mathcal{G}}V, X) = dV(X), \; \mbox{ for  every } X \in
\mathfrak{X}(Q).$$

Using the Riemannian metric $\mathcal{G}$ we can define two
complementary orthogonal projectors ${\mathcal P}\colon TQ\to {\mathcal D}$ and ${\mathcal Q}\colon TQ\to {\mathcal
D}^{\perp},$ with respect to the tangent bundle orthogonal decomposition $\mathcal{D}\oplus\mathcal{D}^{\perp}=TQ$.

In the presence of a constraint distribution $\mathcal{D}$, equation \eqref{ELeq} must be slightly modified as follows. Consider the \textit{nonholonomic connection} $\nabla^{nh}:\mathfrak{X}(Q)\times \mathfrak{X}(Q) \rightarrow \mathfrak{X}(Q)$ defined by (see \cite{bullo} for instance)
\begin{equation}\label{nh:connection}
    \nabla^{nh}_X Y =\nabla_{X}^{\mathcal{G}} Y + (\nabla_{X}^{\mathcal{G}} \mathcal{Q})(Y).
\end{equation}
Then, the trajectories for the nonholonomic mechanical system associated with the Lagrangian \eqref{mechanical:lagrangian} and the distribution $\mathcal{D}$ must satisfy the following equation
\begin{equation}\label{nonholonomic:mechanical:equation}
    \nabla^{nh}_{\dot{q}}\dot{q} + \mathcal{P}(\text{grad}_{\mathcal{G}}V(q(t))) = 0.
\end{equation}

\section{Virtual nonholonomic constraints}\label{sec:controler}
Next, we present the rigorous construction of virtual nonholonomic  constraints. In contrast to the case of standard nonholonomic constraints of the form \eqref{NHconstraint}, the concept of virtual constraint is always associated with a controlled system, rather than with the distribution defined by the constraints. 

Given the Riemannian metric $\mathcal{G}$ on $Q$, we can use its non-degeneracy property to define the musical isomoprhism $\flat:\mathfrak{X}(Q)\rightarrow \Omega^{1}(Q)$ defined by $\flat(X)(Y)=\mathcal{G}(X,Y)$ for any $X, Y \in \mathfrak{X}(Q)$. Also, denote by $\sharp:\Omega^{1}(Q)\rightarrow \mathfrak{X}(Q)$ the inverse musical isomorphism, i.e., $\sharp=\flat^{-1}$.

Given an external force $F^{0}:TQ\rightarrow T^{*}Q$ and a control force $F:TQ\times U \rightarrow T^{*}Q$ of the form
\begin{equation}
    F(q,\dot{q},u) = \sum_{a=1}^{m} u_{a}f^{a}(q)
\end{equation}
where $f^{a}\in \Omega^{1}(Q)$ with $m<n$, $U\subset\mathbb{R}^{m}$ the set of controls and $u_a\in\mathbb{R}$ with $1\leq a\leq m$ the control inputs, consider the associated mechanical control system of the form
\begin{equation}\label{mechanical:control:system}
    \nabla^{\mathcal{G}}_{\dot{q}(t)} \dot{q}(t) =Y^{0}(q(t),\dot{q}(t)) + u_{a}(t)Y^{a}(q(t)),
\end{equation}
with $Y^{0}=\sharp(F^{0})$ and $Y^{a}=\sharp(f^{a})$ the corresponding force vector fields.

Hence, $q$ is the trajectory of a vector field of the form
\begin{equation}\label{SODE}\Gamma(v_{q})=G(v_{q})+u_{a}(Y^{a})_{v_{q}}^{V},\end{equation}
where $G$ is the vector field determined by the unactuated forced mechanical system
\begin{equation*}
    \nabla^{\mathcal{G}}_{\dot{q}(t)} \dot{q}(t) =Y^{0}(q(t),\dot{q}(t))
\end{equation*}
and where the vertical lift of a vector field $X\in \mathfrak{X}(Q)$ to $TQ$ is defined by $$X_{v_{q}}^{V}=\left. \frac{d}{dt}\right|_{t=0} (v_{q} + t X(q)).$$

\begin{defn}
    The distribution $\mathcal{F}\subseteq TQ$ generated by the vector fields  $\sharp(f_{i})$ is called the \textit{input distribution} associated with the mechanical control system \eqref{mechanical:control:system}.
\end{defn}

Now we will define the concept of virtual nonholonomic constraint.

\begin{defn}
A \textit{virtual nonholonomic constraint} associated with the mechanical control system \eqref{mechanical:control:system} is a controlled invariant distribution $\mathcal{D}\subseteq TQ$ for that system, that is, there exists a control function $\hat{u}:\mathcal{D}\rightarrow \mathbb{R}^{m}$ such that the solution of the closed-loop system satisfies $\phi_{t}(\mathcal{D})\subseteq \mathcal{D}$, where $\phi_{t}:TQ\rightarrow TQ$ denotes its flow.
\end{defn}

\begin{rem}A particular example of mechanical control system appearing in applications is determined by a mechanical Lagrangian function $L:TQ\rightarrow \mathbb{R}$. In this case, the control system is given by the controlled Euler-Lagrange equations, i.e., 
\begin{equation}\label{euler:lagrange:system}
    \frac{d}{dt}\left(\frac{\partial L}{\partial \dot{q}}\right)-\frac{\partial L}{\partial q}=F(q,\dot{q},u).
\end{equation}

If the curve $q:I\rightarrow Q$ is a solution of the controlled Euler Lagrange equations \eqref{euler:lagrange:system}, it may be shown that it satisfies the mechanical equation (see \cite{bullo} for instance)
\begin{equation}\label{lagrangian:control:system}
    \nabla^{\mathcal{G}}_{\dot{q}(t)} \dot{q}(t) + \text{grad}_{\mathcal{G}}V(q(t))=u_{a}(t)Y^{a}(q(t)).
\end{equation}
These are the equations of a mechanical control system as in \eqref{mechanical:control:system}, where the force field $Y^{0}$ is simply given by $-\text{grad}_{\mathcal{G}}V(q(t))$. In this case, we call \eqref{lagrangian:control:system} a controlled Lagrangian system.\hfill$\diamond$\end{rem}

\subsection{Relation with previous definitions of virtual nonholonomic constraints}

In previous works, virtual nonholonomic constraints appeared under different definitions. The most general one, comprising every other one as a particular case, is given in \cite{moran2021energy} where a virtual nonholonomic constraint is a set of the form
$$\mathcal{M}=\{(q,p)\in Q\times \mathbb{R}^{n} \ | \ h(q,p)=0\},$$
for which there exists a control law making it invariant under the flow of the closed-loop controlled Hamiltonian equations. This constraint might be rewritten using the cotangent bundle $T^{*}Q$ and $h$ might be seen as a function $h:T^{*}Q\rightarrow \mathbb{R}^{m}$. In addition, $h$ should satisfy $\text{rank } dh(q,p) = m$ for all $(q,p)\in \mathcal{M}$.

Our definition falls under this general definition, for the particular case where the function $h$ is linear on the fibers, i.e., a linear function on the momenta $p_{i}$. In order to see this, we must rewrite the virtual nonholonomic constraints and the control system on the cotangent bundle.

Indeed, consider the Hamiltonian function $H:T^{*}Q \rightarrow \mathbb{R}$ obtained from a Lagrangian function in the following way
$$H(q,p)=p\dot{q}(q,p)-L(q,\dot{q}(q,p)),$$
where $\dot{q}(q,p)$ is a function of $(q,p)$ given by the inverse of the Legendre transformation
$$p=\frac{\partial L}{\partial \dot{q}}.$$
The controlled Hamiltonian equations are given by
$$\dot{q}=\frac{\partial H}{\partial p}, \quad \dot{p}=-\frac{\partial H}{\partial q} + F^{0}(q,\dot{q}(q,p)) + u_{a}f^{a}(q),$$
where $F^{0}$ is an external force map. Now, any distribution $\mathcal{D}\subseteq TQ$ might be defined as the set
$$\mathcal{D}= \{ (q,\dot{q})\in TQ \ | \ \mu^{a}(q)(\dot{q}) = 0\},$$
where $\mu^{a}$ with $1 \leqslant a\leqslant m$ are $m$ linearly independent one-forms. The cotangent version of the distribution is the set
$$\tilde{\mathcal{M}}= \{ (q,p) \ | \ \mu^{a}(q)(\dot{q}(q,p)) = 0 \}. $$
Therefore, we set $$h(q,p)=(\mu^{1}(q)(\dot{q}(q,p)),\cdots, \mu^{m}(q)(\dot{q}(q,p))).$$ We just have to check if $\text{rank } dh = m$. Note that each component of $h$ is linear on fibers if the Lagrangian function (and thus, the corresponding Hamiltonian function) is of mechanical type, i.e., $L=\frac{1}{2}\dot{q}^{T}M\dot{q} - V(q)$, where $M$ is the mass matrix and it represents the Riemanian metric on coordinates, then the Legendre transform is just
$p=M\dot{q}$ and its inverse is $\dot{q}=M^{-1}p$. Therefore,
$$h(q,p)=(\mu^{1}M^{-1}p, \cdots, \mu^{m}M^{-1}p).$$
Hence, the submatrix of the Jacobian formed by the partial derivatives with respect to the momenta $p$ are formed by the rows
$$M^{-1}\mu^{1}, \cdots M^{-1}\mu^{m},$$
which are linearly independent. Thus this submatrix has rank $m$ and this implies that the Jacobian matrix $dh$ has rank greater than $m$. However, since it is formed by $m$ rows, the rank of $dh$ must be exactly $m$ and $\tilde{\mathcal{M}}$ is a virtual nonholonomic constraint according to \cite{moran2021energy} if there is a control law making it invariant.

In summary, in the case that the mechanical control system is described by a mechanical Lagrangian function, our definition of virtual nonholonomic constraint coincides with the one given in \cite{moran2021energy} when we view it in the cotangent bundle. However, their definition is more general than ours since it also comprises nonlinear constraints.

\begin{remark}
    The requirement that the mechanical control system comes from a mechanical Lagrangian is not necessary in order to have equivalence of both definitions but it is at least necessary that we have some way of pushing forward the constraints to the cotangent bundle. This property is usually the regularity of the Lagrangian function, which amounts to have the Legendre transformation as a local diffeomorphism between $TQ$ and $T^{*}Q$.
\end{remark}

\subsection{Examples}

\begin{example}\label{se2:example}
    Consider in $SE(2)\cong \mathbb{R}^{2}\times \mathbb{S}^{1}$ the mechanical Lagrangian function
$$L(x,y,\theta,\dot{x},\dot{y},\dot{\theta})=\frac{m}{2}(\dot{x}^{2}+\dot{y}^{2})+\frac{I\dot{\theta}^{2}}{2}$$
together with the control force
$$F(x,y,\theta,\dot{x},\dot{y},\dot{\theta},u)=u(\sin \theta dx-\cos \theta dy +d\theta).$$
The corresponding controlled Lagrangian system is
\begin{equation*}
    m\ddot{x}=u \sin\theta, \quad m\ddot{y}=-u \cos\theta, \quad I\ddot{\theta}=u
\end{equation*}
and, as we will show, it has the following virtual nonholonomic constraint
$$\sin\theta \dot{x} - \cos\theta \dot{y}=0.$$
The input distribution $\mathcal{F}$ is generated just by one vector field $$Y=\frac{\sin \theta}{m}\frac{\partial}{\partial x}-\frac{\cos \theta}{m}\frac{\partial}{\partial y}+\frac{1}{I}\frac{\partial}{\partial \theta},$$
while the virtual nonholonomic constraint is the distribution $\mathcal{D}$ defined as the set of tangent vectors $v_{q}\in T_{q}Q$  where $\mu(q)(v)=0,$ with $\mu=\sin\theta dx - cos\theta dy$. Thus, we may write it as
$$\mathcal{D}=\hbox{span}\Big{\{} X_{1}=\cos \theta\frac{\partial}{\partial x} + \sin\theta \frac{\partial}{\partial y},\, X_{2}=\frac{\partial}{\partial \theta}\Big{ \}}.$$
We may check that $\mathcal{D}$ is controlled invariant for the controlled Lagrangian system above. In fact, the control law
$$\hat{u}(x,y,\theta,\dot{x},\dot{y},\dot{\theta})=-m\dot{\theta}(\cos\theta \dot{x} +\sin \theta \dot{y})$$
makes the distribution invariant under the closed-loop system, since in this case, the dynamical vector field arising from the controlled Euler-Lagrange equations given by
$$\Gamma = \dot{x}\frac{\partial}{\partial x} + \dot{y}\frac{\partial}{\partial y} + \dot{\theta}\frac{\partial}{\partial \theta} + \frac{\hat{u}\sin \theta}{m}\frac{\partial}{\partial \dot{x}} - \frac{\hat{u}\cos \theta}{m}\frac{\partial}{\partial \dot{y}} + \frac{\hat{u}}{I}\frac{\partial}{\partial \dot{\theta}}$$
is tangent to $\mathcal{D}$. This is deduced from the fact that $\Gamma(\sin\theta \dot{x} - \cos\theta \dot{y})=0$.\hfill$\diamond$
\end{example}

\begin{example}\label{disk:example}
    Consider in $\mathbb{R}^{2}\times \mathbb{S}^{1} \times \mathbb{S}^{1}$ the mechanical Lagrangian function
$$L(x,y,\theta,\varphi,\dot{x},\dot{y},\dot{\theta},\dot{\varphi})=\frac{m}{2}(\dot{x}^{2}+\dot{y}^{2})+\frac{I\dot{\theta}^{2}}{2} + \frac{J\dot{\varphi}^{2}}{2}$$
together with the control force \begin{align*}F(x,y,\theta,\varphi,\dot{x},\dot{y},\dot{\theta},\dot{\varphi},u)=&u_{1}(dx-\cos \varphi d\theta + d\varphi)\\ &+ u_{2}(dy-\sin \varphi d\theta + d\varphi).\end{align*}
The controlled Lagrangian system is then
\begin{equation*}
    m\ddot{x}=u_{1},\, m\ddot{y}=u_{2},\, I\ddot{\theta}=-u_{1}\cos \varphi - u_{2} \sin \varphi,\, J\ddot{\varphi}=u_{1} + u_{2}.
\end{equation*}
The virtual nonholonomic constraints associated to this system are defined by the following equations
\begin{equation*}
    \dot{x}=\dot{\theta}\cos \varphi, \quad \dot{y} = \dot{\theta}\sin\varphi.
\end{equation*}
Therefore, the input distribution $\mathcal{F}$ is the set
\begin{equation*}
    \begin{split}
        \mathcal{F}=\hbox{span}\Big{\{} Y^{1} = & \frac{1}{m}\frac{\partial}{\partial x}-\frac{\cos \varphi}{I}\frac{\partial}{\partial \theta}+\frac{1}{J}\frac{\partial}{\partial \varphi}, \\
        & Y^{2}= \frac{1}{m}\frac{\partial}{\partial y}-\frac{\sin \varphi}{I}\frac{\partial}{\partial \theta}+\frac{1}{J}\frac{\partial}{\partial \varphi} \Big{\}}, 
    \end{split}
\end{equation*} and the constraint distribution $\mathcal{D}$ is defined by the 1-forms $\mu^{1}=dx - \cos \varphi d\theta$ and $\mu^{2}=dy - \sin \varphi d\theta$, thus
$$\mathcal{D}=\Big{\{} X_{1}=\cos \varphi\frac{\partial}{\partial x} + \sin\varphi \frac{\partial}{\partial y} + \frac{\partial}{\partial \theta},\, X_{2}= \frac{\partial}{\partial \varphi}\Big{\}}.$$
We may verify, using a similar argument as Example \ref{se2:example}, that $\mathcal{D}$ is in fact controlled invariant under the control law
\begin{equation*}
        \hat{u}_{1}= -m\dot{\theta}\dot{\varphi}\sin \varphi, \quad
        \hat{u}_{2}=  m\dot{\theta}\dot{\varphi} \cos \varphi.
\end{equation*}\hfill$\diamond$
\end{example}

\begin{example}
   Let us look at an example of a mechanical control system which is not a Lagrangian system. Consider again the mechanical control system proposed in Example \ref{se2:example} but now with an additional damping term determined by the vector fiel
   $Y^{0}= -\frac{\gamma}{m}(\dot{x}dx + \dot{y} dy)$,
   where $\gamma>0$ is a damping constant. The mechanical control system has the following equations of motion
   \begin{equation*}
    m\ddot{x}=u \sin\theta - \gamma \dot{x}, \quad m\ddot{y}=-u \cos\theta - \gamma \dot{y}, \quad I\ddot{\theta}=u.
    \end{equation*}
    It is not difficult to check that the control law
    $$\hat{u}(x,y,\theta,\dot{x},\dot{y},\dot{\theta})=-m\dot{\theta}(\cos\theta \dot{x} +\sin \theta \dot{y})$$
    still makes the distribution invariant under the flow of the closed-loop system.
   \hfill$\diamond$ 
\end{example}

\subsection{Existence and uniqueness of a feedback control making the constraints invariant}

It is often very useful if we have conditions under which we are guaranteed that a distribution $\mathcal{D}$ is controlled invariant for the controlled Lagrangian system \eqref{lagrangian:control:system}. The next result not only states the existence of a control function making $\mathcal{D}$ invariant, but it also states that it is unique. In the following, two distributions $\mathcal{A}_{1}$ and $\mathcal{A}_{2}$ on the manifold $Q$ are said to be transversal if they are complementary, in the sense that $TQ=\mathcal{A}_{1}\oplus\mathcal{A}_{2}$.

\begin{theorem}\label{main:theorem}
If the distribution $\mathcal{D}$ and the control input distribution $\mathcal{F}$ are transversal, then there exists a unique control function making the distribution a virtual nonholonomic constraint associated with the mechanical control system \eqref{mechanical:control:system}.
\end{theorem}

\begin{proof}
 
    Suppose that $TQ=\mathcal{D}\oplus \mathcal{F}$ and that  trajectories of the contol system \eqref{mechanical:control:system} may be written as the integral curves of the vector field $\Gamma$ defined by \eqref{SODE}. For each $v_{q}\in \mathcal{D}_{q}$, we have that $$\Gamma(v_{q})\in T_{v_{q}}(TQ)=T_{v_{q}}\mathcal{D}\oplus \hbox{span}\Big{\{}(Y^{a})_{v_{q}}^{V} \Big{\}},$$ with $Y^{a}=\sharp(f^{a})$. Using the uniqueness decomposition property arising from transversality, we conclude there exists a unique vector $\tau^{*}(v_{q})=(\tau_{1}^{*}(v_{q}),\cdots, \tau_{m}^{*}(v_{q}))\in \mathbb{R}^{m}$ such that $$\Gamma(v_{q})+\tau_{a}^{*}(v_{q})(Y^{a})_{v_{q}}^{V}\in T_{v_{q}}\mathcal{D}.$$ If $\mathcal{D}$ is defined by $m$ constraints of the form $\phi^{b}(v_{q})=0$, $1\leq b\leq m$, then the condition above may be rewritten as $$d\phi^{b}(\Gamma(v_{q})+\tau_{a}^{*}(v_{q})(Y^{a})_{v_{q}}^{V})=0,$$ which is equivalent to $$\tau_{a}^{*}(v_{q})d\phi^{b}((Y^{a})_{v_{q}}^{V})=-d\phi^{b}(\Gamma(v_{q})).$$ Note that, the equation above is a linear equation of the form $A(v_{q})\tau=b(v_{q})$, where $b(v_{q})$ is the vector $(-d\phi^{1}(\Gamma(v_{q})), \dots, -d\phi^{m}(\Gamma(v_{q})))\in \mathbb{R}^{m}$ and $A(v_{q})$ is the $m\times m$ matrix with entries $A^{b}_{a}(v_{q})=d\phi^{b}((Y^{a})_{v_{q}}^{V})=\mu^{b}(q)(Y^{a})$, where the last equality may be deduced by computing the expressions in local coordinates. That is, if $(q^{i} \dot{q}^{i})$ are natural bundle coordinates for the tangent bundle, then
    \begin{equation*}
        \begin{split}
            d\phi^{b}((Y^{a})_{v_{q}}^{V}) & = \left(\frac{\partial \mu^{b}_{i}}{\partial q^{j}}\dot{q}^{i}dq^{j} + \mu^{b}_{i}d\dot{q}^{i}\right)\left(Y^{a,k}\frac{\partial}{\partial \dot{q}^{k}}\right) \\
            & = \mu^{b}_{i}Y^{a,i} = \mu^{b}(q)(Y^{a}).
        \end{split}
    \end{equation*}
    In addition, $A(v_{q})$ has full rank, since its columns are linearly independent. In fact suppose that
    \begin{equation*}
        c_{1}\begin{bmatrix} \mu^{1}(Y^{1}) \\
        \vdots \\
        \mu^{m}(Y^{1}) \end{bmatrix} + \cdots + c_{m}\begin{bmatrix} \mu^{1}(Y^{m}) \\
        \vdots \\
        \mu^{m}(Y^{m}) \end{bmatrix}= 0,
    \end{equation*}
    which is equivalent to
    \begin{equation*}
        \begin{bmatrix} \mu^{1}(c_{1}Y^{1}+\cdots + c_{m}Y^{m}) \\
        \vdots \\
        \mu^{m}(c_{1}Y^{1}+\cdots + c_{m}Y^{m}) \end{bmatrix}=0.
    \end{equation*}
    However, by transversality we have $\mathcal{D}\cap \mathcal{F} = \{0\}$ which implies that $c_{1}Y^{1}+\cdots + c_{m}Y^{m}=0$. Since $\{Y_{i}\}$ are linearly independent we conclude that $c_{1}=\cdots=c_{m}=0$ and $A$ has full rank. But, since $A$ is an $m\times m$ matrix, and $\mathcal{D}$ is a regular distribution, it must be invertible. Therefore, there is a unique vector $\tau^{*}(v_{q})$ satisfying the matrix equation and $\tau^{*}:\mathcal{D}\rightarrow \mathbb{R}^{m}$ is smooth since it is the solution of a matrix equation depending smoothly on $v_{q}$.
\end{proof}

\begin{remark}
    Note that in Examples \ref{se2:example} and \ref{disk:example}, the constraint distribution $\mathcal{D}$ and the control input distribution $\mathcal{F}$ are transversal. Thus the control laws obtained in there are unique by Theorem \ref{main:theorem}.\hfill$\diamond$
\end{remark}

The transversality condition is essential in order to have existence and uniqueness of the control law making the constraint distribution control invariant. If they are not transversal then a control law making $\mathcal{D}$ control invariant may not exist or may not be unique as we will see in the next examples.

\begin{example}[Non-existence]
    Consider the Lagrangian function $L$ and the distribution $\mathcal{D}$ given in Example \ref{se2:example}, but now let the control force be
    $$F(x,y,\theta,\dot{x},\dot{y},\dot{\theta},u) = u (\cos \theta dx + \sin \theta dy),$$
    so that the controlled Lagrangian system is now
    \begin{equation*}
        m\ddot{x}=u \cos\theta, \quad m\ddot{y}=u \sin\theta, \quad I\ddot{\theta}=0.
    \end{equation*}
    Note that, in this case, the control input distribution $\mathcal{F}$ is generated by the vector field $\displaystyle{Y=\frac{\cos \theta}{m}\frac{\partial}{\partial x} + \frac{\sin \theta}{m}\frac{\partial}{\partial y}}$. Hence, $\mathcal{F}\subseteq \mathcal{D}$.
    
    Suppose that a control law $\hat{u}$ making the distribution control invariant exists. Differentiating the constraints, we get
    $$\cos\theta \dot{x} + \sin\theta \ddot{x} + \sin \theta \dot{y} - \cos\theta \ddot{y} = 0,$$
    and substituting by the closed-loop system we get
    $$0 = \cos\theta \dot{x} + \frac{\hat{u} \sin\theta\cos\theta}{m} + \sin \theta \dot{y} -  \frac{\hat{u} \sin\theta \cos\theta}{m},$$
    which is satisfied only when $\cos\theta \dot{x} + \sin \theta \dot{y} =0$. Therefore, there is no control law $\hat{u}$ making the distribution control invariant.\hfill$\diamond$
    \end{example}
    
    \begin{example}[Non-uniqueness]
        Consider again the situation given in Example \ref{se2:example} but now with the control force
        \begin{equation*}
            \begin{split}
                F(x,y,\theta,\dot{x},\dot{y},\dot{\theta},u) & = u_{1} (\sin \theta dx-\cos \theta dy +d\theta) \\
                &  + u_{2}(\sin \theta dx-\cos \theta dy).
            \end{split}
        \end{equation*}
        In this case, we have that $TQ=\mathcal{D} + \mathcal{F}$ but $\mathcal{D}\cap \mathcal{F} \neq \{0\}$.
        Two examples of control laws making $\mathcal{D}$ control invariant are
        $$\hat{u}_{1} = -m\dot{\theta}(\cos\theta \dot{x} +\sin \theta \dot{y}), \quad \hat{u}_{2} = 0$$
        and
        $$\hat{u}_{1} = 0, \quad \hat{u}_{2}=-m\dot{\theta}(\cos\theta \dot{x} +\sin \theta \dot{y}).$$\hfill$\diamond$
    \end{example}

\section{The induced constrained connection}\label{sec4}

From now on suppose that the distribution $\mathcal{D}$ describing the virtual nonholonomic constraints and the input distribution $\mathcal{F}$ are transversal. Therefore, the projections $P_{\mathcal{F}}:TQ\rightarrow \mathcal{F}$ and $P_{\mathcal{D}}:TQ\rightarrow \mathcal{D}$ associated to the direct sum are well-defined.

The \textit{induced constrained connection} associated to the distribution $\mathcal{D}$ and the input distribution $\mathcal{F}$ is given by
\begin{equation}\label{virtual:nh:connection}
    \overset{c}{\nabla}_{X} Y = \nabla^{\mathcal{G}}_X Y + (\nabla_{X}^{\mathcal{G}}P_{\mathcal{F}})(Y),
\end{equation}
where $\nabla^{\mathcal{G}}$ is the Levi-Civita connection associated with the Riemannian metric $\mathcal{G}$. The induced constrained connection is a linear connection on $Q$ with the special property that $\mathcal{D}$ is geodesically invariant for $\overset{c}{\nabla}$, i.e., if a geodesic of $\overset{c}{\nabla}$ starts on $\mathcal{D}$ then it stays in $\mathcal{D}$ for all time (see \cite{lewis1998affine}).

We have the following useful lemma that we will use later on.

\begin{lemma}\label{constrained:property}
    If $X,Y \in \Gamma(\tau_{\mathcal{D}})$ then
    $$\overset{c}{\nabla}_{X} Y = P_{\mathcal{D}}(\nabla^{\mathcal{G}}_X Y).$$
\end{lemma}

\begin{proof}
If $X,Y \in \Gamma(\tau_{\mathcal{D}})$ we have that
    \begin{equation*}
        \begin{split}
            \overset{c}{\nabla}_{X} Y = & \nabla^{\mathcal{G}}_{X} Y + (\nabla^{\mathcal{G}}_{X}P_{\mathcal{F}})(Y) \\
             = & \nabla^{\mathcal{G}}_{X} Y + \nabla^{\mathcal{G}}_{X}(P_{\mathcal{F}}(Y))- P_{\mathcal{F}}(\nabla^{\mathcal{G}}_{X} Y),
        \end{split}
    \end{equation*}
    where we have used the definition of covariant derivative of a map of the form $T:TQ\rightarrow TQ$ in the last equality. Noting that $P_{\mathcal{F}}(Y)=0$ since $Y$ is a section of $\tau_{\mathcal{D}}$, we conclude that $\overset{c}{\nabla}_{X} Y = P_{\mathcal{D}}(\nabla^{\mathcal{G}}_{X} Y).$
\end{proof}

The last lemma implies in particular that $\overset{c}{\nabla}$ is well-defined as a connection on sections of $\tau_{\mathcal{D}}$ in the sense that the restriction $\overset{c}{\nabla}|_{\Gamma(\tau_{\mathcal{D}})\times \Gamma(\tau_{\mathcal{D}})}$ takes values also on $\Gamma(\tau_{\mathcal{D}})$. However, as the following lemma shows the constrained connection is not symmetric, in general.

\begin{lemma}\label{symmetric:lemma}
    If the constrained connection $\overset{c}{\nabla}$ is symmetric then the constraint distribution $\mathcal{D}$ is integrable.
\end{lemma}
\begin{proof}
    The torsion of the constrained connection is given by
    $$T^{c}(X,Y) = \overset{c}{\nabla}_{X}Y - \overset{c}{\nabla}_{Y} X - [X,Y].$$
    
    Suppose that $X, Y \in \Gamma(\tau_{\mathcal{D}})$. In this case
    \begin{equation*}
        \begin{split}
            T^{c}(X,Y) & = P_{\mathcal{D}}(\nabla^{\mathcal{G}}_{X} Y-\nabla^{\mathcal{G}}_{Y} X) - [X,Y] \\
            & = P_{\mathcal{D}}([X,Y]) - [X,Y] \\
            & = -P_{\mathcal{F}}([X,Y]),
        \end{split}
    \end{equation*}
    where we used the fact that $\nabla^{\mathcal{G}}$ is symmetric in the first equality. It is clear now that if $\overset{c}{\nabla}$ is symmetric then $[X,Y]$ must be a section of $\mathcal{D}$, which implies that $\mathcal{D}$ is integrable.
\end{proof}

\begin{remark}
    Lemma \ref{symmetric:lemma} was also proved in \cite{lewis1998affine}, however we provided here an alternative simple proof in order to keep the discussion as much self-contained as possible.\hfill$\diamond$
\end{remark}

In the following, we characterize the closed-loop dynamics as solutions of the mechanical system associated with the induced constrained connection.

\begin{theorem}
    A curve $q:I\rightarrow Q$ is a trajectory of the closed-loop system for the Lagrangian control system \eqref{lagrangian:control:system} making $\mathcal{D}$ invariant if and only if it satifies
    \begin{equation}\label{constrained:equation}
        \overset{c}{\nabla}_{\dot{q}(t)} \dot{q}(t) + P_{\mathcal{D}}(\text{grad}_{\mathcal{G}} V(q(t)))=0.
    \end{equation}
\end{theorem}

\begin{proof}
    If $q:I\rightarrow Q$ is a trajectory of the closed-loop system for \eqref{lagrangian:control:system} with $\dot{q}(t)\in \mathcal{D}_{q(t)}$ then it satisfies
    $$\nabla^{\mathcal{G}}_{\dot{q}(t)} \dot{q}(t) + \text{grad}_{\mathcal{G}}V(q(t))=\hat{u}_{a}(t)Y^{a}(q(t)),$$
    where $\hat{u}:\mathcal{D}\rightarrow \mathbb{R}^{m}$ is the unique control law making $\mathcal{D}$ invariant. Attending to the fact that $\dot{q}(t)\in \mathcal{D}_{q(t)}$ we have that
    \begin{equation*}
    \begin{split}
            \overset{c}{\nabla}_{\dot{q}(t)} \dot{q}(t) = & P_{\mathcal{D}}(\nabla^{\mathcal{G}}_{\dot{q}(t)} \dot{q}(t)) \\
             = & -P_{\mathcal{D}}(\text{grad}_{\mathcal{G}}V(q(t))) + P_{\mathcal{D}}(\hat{u}_{a}(t)Y^{a}(q(t))) \\
             = & - P_{\mathcal{D}}(\text{grad}_{\mathcal{G}}V(q(t))),
        \end{split}
    \end{equation*}
    where we have used Lemma \ref{constrained:property} in the first equality and $P_{\mathcal{D}}(Y^{a})=0$ in the last one.
    
    Conversely, if the curve $q$ satisfies \eqref{constrained:equation}, we have
    $$P_{\mathcal{D}}(\nabla^{\mathcal{G}}_{\dot{q}(t)} \dot{q}(t) + \text{grad}_{\mathcal{G}}V(q(t)))=0,$$
    where we used Lemma \ref{constrained:property}. Since $\ker P_{\mathcal{D}} = \mathcal{F}$, there exist $u=(u_{1}, \cdots, u_{m})\in \mathbb{R}^{m}$ such that
    $$\nabla^{\mathcal{G}}_{\dot{q}(t)} \dot{q}(t) + \text{grad}_{\mathcal{G}}V(q(t)) = u_{a}Y^{a}.$$
    By Theorem \ref{main:theorem}, we conclude that $u=\hat{u}$, since the control law making $\mathcal{D}$ invariant is unique.
\end{proof}

\begin{remark}
    Suppose $\mathcal{D}$ is an integrable distribution and assume $\mathcal{C}$ is a maximal integrable manifold of $\mathcal{D}$. If $\overset{h}{\nabla}$ denotes the holonomic connection on $\mathcal{C}$ defined in \cite{consolini2015induced} (see also \cite{consolini2018coordinate}), as $$\overset{h}{\nabla}_{X}Y = P_{\mathcal{D}}(\nabla^{\mathcal{G}}_{X} Y), \quad X,Y \in \mathfrak{X}(\mathcal{C}),$$ then Lemma \ref{constrained:property} implies that the two connections are the same when $\overset{c}{\nabla}$ is restricted to vector fields on $\mathcal{C}$.\hfill$\diamond$
\end{remark}

\subsection{The constrained connection in coordinates}

In this section we will compute the Christoffel symbols of the induced connection. Given any coordinate chart $(q^{i})$ on $Q$ the Christoffel symbols are determined by the values of the connection taken over the standard basis of the tangent space $\{\frac{\partial}{\partial q^{1}}, \cdots, \frac{\partial}{\partial q^{n}} \}$. It is not difficult to prove the following useful expression
\begin{equation*}
    \overset{c}{\nabla}_{\frac{\partial}{\partial q^{i}}}\frac{\partial}{\partial q^{j}}=P_{\mathcal{D}}\left(\nabla^{\mathcal{G}}_{\frac{\partial}{\partial q^{i}}}\frac{\partial}{\partial q^{j}}\right) + \nabla^{\mathcal{G}}_{\frac{\partial}{\partial q^{i}}}\left(P_{\mathcal{F}}\left(\frac{\partial}{\partial q^{j}}\right)\right).
\end{equation*}

\begin{example}
    Consider once again the control system given in Example  \ref{se2:example}. The Levi-Civita connection $\nabla^{\mathcal{G}}$ associated with this system has vanishing Christoffel symbols. Considering the coordinates $q=(x,y,\theta)$ on $SE(2)$, we have that
    $$\overset{c}{\nabla}_{\frac{\partial}{\partial q^{i}}}\frac{\partial}{\partial q^{j}}=\nabla^{\mathcal{G}}_{\frac{\partial}{\partial q^{i}}}\left(P_{\mathcal{F}}\left(\frac{\partial}{\partial q^{j}}\right)\right).$$
    
    Note that the natural coordinate vector fields for $SE(2)$ may be decomposed in a unique way, under the direct sum $\mathcal{D}\oplus \mathcal{F}$, and this decomposition is given by
    \begin{equation*}
        \begin{split}
            & \frac{\partial}{\partial x} = \cos \theta X_{1} - \frac{m \sin \theta}{I} X_{2} + m \sin \theta Y, \\
            & \frac{\partial}{\partial y} = \sin \theta X_{1} + \frac{m \cos \theta}{I} X_{2} - m \cos \theta Y, \\
            & \frac{\partial}{\partial \theta} = X_{2}.
        \end{split}
    \end{equation*}
    
    Hence, we obtain the following non-vanishing Christoffel symbols for the constrained connection $\overset{c}{\nabla}$
    \begin{equation*}
        \begin{split}
            & \Gamma_{\theta x}^{x}= 2\sin\theta\cos\theta, \,\quad \quad \Gamma_{\theta y}^{x}= \sin^{2} \theta - \cos^{2} \theta, \\
            & \Gamma_{\theta x}^{y}= \sin^{2} \theta - \cos^{2} \theta,\,\quad \Gamma_{\theta y}^{y}= -2\sin\theta\cos\theta,\\
              & \Gamma_{\theta x}^{\theta}=\frac{m \cos\theta}{I}, \quad\,\,\,\,\,\,\,\,\qquad \Gamma_{\theta y}^{\theta}=\frac{m \sin\theta}{I}.
        \end{split}
    \end{equation*}
   If we introduce the coordinates $q=(x,y,\theta,\varphi)$ in Example \ref{disk:example} and following the same reasoning we get
   \begin{small} 
    \begin{equation*}
        \begin{split}
            P_{\mathcal{F}}\left(\frac{\partial}{\partial x}\right) & = \frac{I J m + J m^{2} \sin^{2}{\left(\varphi \right)}}{L(\varphi)} Y^{1} - \frac{J m^{2} \sin{\left(\varphi \right)} \cos{\left(\varphi \right)}}{L(\varphi)} Y^{2} \\
            P_{\mathcal{F}}\left(\frac{\partial}{\partial y}\right) & = \frac{I m - J m^{2} \sin{\left(\varphi \right)} \cos{\left(\varphi \right)}}{L(\varphi)}Y^{1} + \frac{- I m + J m^{2} \cos^{2}{\left(\varphi \right)}}{L(\varphi)} Y^{2} \\
            P_{\mathcal{F}}\left(\frac{\partial}{\partial \theta}\right) & = \frac{- I J m \cos{\left(\varphi \right)} - I m \sin{\left(\varphi \right)}}{L(\varphi)}Y^{1} + \frac{I m \sin{\left(\varphi \right)}}{L(\varphi)} Y^{2} \\
            P_{\mathcal{F}}\left(\frac{\partial}{\partial \varphi}\right) & = \frac{- I J - J m \sin^{2}{\left(\varphi \right)}}{L(\varphi)} Y^{1} + \frac{J m \sin{\left(\varphi \right)} \cos{\left(\varphi \right)}}{L(\varphi)} Y^{2},
        \end{split}
    \end{equation*}
    \end{small}with $L(\varphi) = - I + J m \cos^{2}{\left(\varphi \right)} - m \sin^{2}{\left(\varphi \right)} + m \sin{\left(\varphi \right)} \cos{\left(\varphi \right)}$. In addition, the non-vanishing Christoffel symbols are given in Appendix \ref{appendix}.
\end{example}

\section{Existence of a nonholonomic Lagrangian structure for the dynamics on $\mathcal{D}$}\label{sec5}

The next proposition shows that if the input distribution is orthogonal to the virtual nonholonomic constraint distribution then the constrained dynamics is precisely the nonholonomic dynamics with respect to the original Lagrangian function.

\begin{proposition}\label{orthogonal:input:distribution}
If the input distribution $\mathcal{F}$ is orthogonal to the virtual constraint distribution $\mathcal{D}$ with respect to the metric $\mathcal{G}$, then the trajectories of the constrained mechanical system \eqref{constrained:equation} are the nonholonomic equations of motion.
\end{proposition}
\begin{proof}
 
 If $\mathcal{F}=\mathcal{D}^{\bot}$, then the projectors $P_{\mathcal{D}}$ and $\mathcal{P}$ coincide (as well as the projectors $P_{\mathcal{F}}$ and $\mathcal{Q}$). Thus, the constrained connection $\overset{c}{\nabla}$ is precisely the nonholonomic connection $\nabla^{nh}$. This implies that the trajectories of the constrained connection are nonholonomic trajectories.
     \end{proof}
     
\begin{remark}
    The fact that $\mathcal{F}=\mathcal{D}^{\bot}$ is independent of the chosen metric. Once you fix the control force $F$ and let the control input distribution be obtained using the musical isomorphism $\sharp$ as in Section \ref{sec:controler}, then $\mathcal{F}$ is orthogonal to $\mathcal{D}$ if and only if $f^{a}\in \mathcal{D}^{o}$, for $a=1, \cdots, m$.\hfill$\diamond$
\end{remark}

Although the orthogonality condition $\mathcal{F}=\mathcal{D}^{\bot}$ is sufficient in order for the constrained dynamics to be the nonholonomic dynamics, it is not necessary as the following result shows.

\begin{proposition}
Suppose there exists a modified potential function $\tilde{V}$ satisfying
\begin{equation}
    \mathcal{P}(\text{grad}_{\mathcal{G}} \tilde{V}) = P_{\mathcal{\mathcal{D}}}(\text{grad}_{\mathcal{G}} V).
\end{equation}
Then the nonholonomic trajectories with respect to $(\mathcal{G},\tilde{V}, \mathcal{D})$ coincide with the constrained dynamics \eqref{constrained:equation} if and only if $\nabla_{X}^{\mathcal{G}}\mathcal{Q}(X)=\nabla_{X}^{\mathcal{G}}P_{\mathcal{F}}(X)$ for all $X\in \Gamma(\mathcal{D})$.
\end{proposition}

\begin{proof}
    It is not difficult to see that $\nabla_{X}^{\mathcal{G}}\mathcal{Q}(X)=\nabla_{X}^{\mathcal{G}}P_{\mathcal{F}}(X)$ if and only if the two connections satisfy $\overset{c}{\nabla}_{X} X = \nabla_{X}^{nh} X$.
    Therefore, the equation
    $$\overset{c}{\nabla}_{\dot{q}(t)} \dot{q}(t) + P_{\mathcal{D}}(\text{grad}_{\mathcal{G}} V(q(t)))=0$$
    holds if and only if
    $${\nabla}^{nh}_{\dot{q}(t)} \dot{q}(t) + \mathcal{P}(\text{grad}_{\mathcal{G}} \tilde{V}(q(t)))=0$$
    also holds.
    
    Conversely, if the trajectory $q(t)$ satisfies both equation, then $${\nabla}^{nh}_{\dot{q}(t)} \dot{q}(t) = \overset{c}{\nabla}_{\dot{q}(t)} \dot{q}(t)$$
    is also satisfied. Using tensoriality of the difference tensor
    $$D(X,Y)=\overset{c}{\nabla}_{X} Y - \nabla_{X}^{nh} Y,$$
    we may evaluate $D$ point-wise so that
    $$D(X_{q},X_q) = (\overset{c}{\nabla}_{X} X - \nabla_{X}^{nh} X) (q).$$
    Choosing the trajectory $q(t)$ with initial point $q$ and initial veclocity $X_{q}\in \mathcal{D}_{q}$, which is always possible thanks to the existence and uniqueness theorem for ODE, we deduce that $D(X_{q},X_{q})=0$ for any $X_{q}\in \mathcal{D}_{q}$. Hence, $D(X,X)=0$ which is equivalent to $\nabla_{X}^{\mathcal{G}}\mathcal{Q}(X)=\nabla_{X}^{\mathcal{G}}P_{\mathcal{F}}(X)$.
\end{proof}

In the absence of a potential function, i.e., $V=0$, the nonholonomic trajectories coincide with the constrained dynamics if and only if $\nabla_{X}^{\mathcal{G}}\mathcal{Q}(X)=\nabla_{X}^{\mathcal{G}}P_{\mathcal{F}}(X)$ for any $X\in \Gamma(\tau_{\mathcal{D}})$.

Note that the previous characterization of when both dynamics have the same trajectories may be equivalently written as $$\mathcal{P}(\nabla_{X}^{\mathcal{G}} X) = P_{\mathcal{D}}(\nabla_{X}^{\mathcal{G}} X) \text{ or } \mathcal{Q}(\nabla_{X}^{\mathcal{G}} X) = P_{\mathcal{F}}(\nabla_{X}^{\mathcal{G}} X)$$
for any $X\in \Gamma(\tau_{\mathcal{D}})$.

\begin{corollary}
If the geodesic vector field associated with $\nabla^{\mathcal{G}}$ is tangent to $\mathcal{D}$, then the nonholonomic trajectories coincide with the constrained geodesics and they are both the geodesics of $\nabla^{\mathcal{G}}$ with initial velocity in $\mathcal{D}$.
\end{corollary}

\begin{proof}
    We just have to establish that the geodesic vector field associated with $\nabla^{\mathcal{G}}$ is tangent to $\mathcal{D}$ if and only if $\nabla^{\mathcal{G}}_{X} X \in \Gamma(\tau_{\mathcal{D}})$ for every $X \in \Gamma(\tau_{\mathcal{D}})$. Then this is equivalent to $\mathcal{Q}(\nabla_{X}^{\mathcal{G}} X) = 0$ and also to $P_{\mathcal{F}}(\nabla_{X}^{\mathcal{G}} X)=0$. Hence, by the previous result, the geodesics with initial velocity in $\mathcal{D}$ of $\nabla^{nh}$ coincide with the geodesics with initial velocity in $\mathcal{D}$ of $\overset{c}{\nabla}$.
    
    Now, $\nabla^{\mathcal{G}}_{X} X \in \Gamma(\tau_{\mathcal{D}})$ for every $X \in \Gamma(\tau_{\mathcal{D}})$ if and only if $\mathcal{D}$ is geodesically invariant with respect to $\nabla^{\mathcal{G}}$ (see \cite{lewis1998affine}, Theorem 5.4). Using standard results on differential geometry, $\mathcal{D}$ is geodesically invariant with respect to $\nabla^{\mathcal{G}}$ if and only if the geodesic vector field associated with $\nabla^{\mathcal{G}}$ is tangent to $\mathcal{D}$.
\end{proof}

\begin{rem}
One important feature of the theory of virtual holonomic constraints presented in \cite{consolini2018coordinate} is that if the induced connection has the same trajectories as the Levi-Civita connection with respect to the induced metric on the constraint submanifold $\mathcal{C}\subseteq Q$, then the two connections are the same. However, its argument relies on the fact that the induced connection is symmetric. Therefore, the result does not follow in the nonholonomic case whenever the distribution is not integrable.
\end{rem}

The next example illustrates Proposition \ref{orthogonal:input:distribution}.

\begin{example}\label{Chaplygin:sleigh}
    Consider the Chaplygin sleigh, a celebrated example of a nonholonomic mechanical system evolving on the configuration manifold $SE(2)$ with Lagrangian function as in Example \ref{se2:example} but now we consider the control force
    $$F(x,y,\theta,\dot{x},\dot{y},\dot{\theta},u)=u(\sin \theta dx-\cos \theta dy).$$
    The corresponding controlled Lagrangian system is
    \begin{equation*}
        m\ddot{x}=u \sin\theta, \quad m\ddot{y}=-u \cos\theta, \quad I\ddot{\theta}=0.
    \end{equation*} The input distribution $\mathcal{F}$ is generated just by one vector field $$Y=\frac{\sin \theta}{m}\frac{\partial}{\partial x}-\frac{\cos \theta}{m}\frac{\partial}{\partial y},$$
    while the virtual nonholonomic constraint is the same distribution $\mathcal{D}$ as in Example \ref{se2:example}. We may check that the control law
    $$\hat{u}(x,y,\theta,\dot{x},\dot{y},\dot{\theta})=-m\dot{\theta}(\cos\theta \dot{x} +\sin \theta \dot{y})$$
    makes the distribution invariant under the closed-loop system. In addition, by Proposition \ref{orthogonal:input:distribution} the resulting system is precisely the nonholonomic equation \eqref{nonholonomic:mechanical:equation} for the Chaplygin system, since the input distribution spanned by $Y$ is orthogonal to the virtual nonholonomic constraints. \hfill $\diamond$
\end{example}

\begin{remark}
    There are plenty of ways to impose a virtual nonholonomic constraint on a mechanical control system in order to obtain a nonholonomic system. In the last example, one could choose the control force to be
    $$F(x,y,\theta,\dot{x},\dot{y},\dot{\theta},u_{1},u_{2})=u_{1} \sin \theta dx + u_{2}\cos \theta dy$$
    and the corresponding controlled Lagrangian system would be
    \begin{equation*}
        m\ddot{x}=u_{1} \sin\theta, \quad m\ddot{y}=u_{2} \cos\theta, \quad I\ddot{\theta}=0.
    \end{equation*}
    Then, the control law
    $$\hat{u}_{1}(x,y,\theta,\dot{x},\dot{y},\dot{\theta})=-m\dot{\theta}(\cos\theta \dot{x} +\sin \theta \dot{y}), \quad \hat{u}_{2} = -\hat{u}_{1}$$
    makes the the closed-loop system coincide again with the nonholonomic equations for the Chaplygin system. Note that the input distribution is now generated by the vector fields $Y^{1}=\frac{\sin \theta}{m}\frac{\partial}{\partial x}$ and $Y^{2} = \frac{\cos \theta}{m}\frac{\partial}{\partial y}.$ Since they do not generate a transversal distribution to $\mathcal{D}$, we should not expect the control law to be unique. \hfill $\diamond$
\end{remark}



\begin{remark}

    Under the conditions of Proposition \ref{orthogonal:input:distribution}, certain mechanical control systems may be driven to desired stable trajectories by imposing virtual nonholonomic constraints and using the proper control force.

    For instance, for the mechanical control system appearing in Example \ref{Chaplygin:sleigh}, we may drive the system to an asymptotically stable trajectory characterized by $\dot{\theta}=0$. Indeed, by defining the variables $v = \dot{x} \cos \theta + \dot{y} \sin \theta$ and $\omega=\dot{\theta}$, the equations of motion of the Chaplygin sleigh might be written as
    $$\dot{\omega}=-\frac{ma}{I+ma^{2}}v \omega, \quad \dot{v} = a \omega^{2},$$
    for which the points with $\omega=0$ are equilibria. Moreoever, from a stability analysis we deduce that the system exhibits asymptotic stability along 
    a trejectory defined by $\omega=0$.

    Nonholonomic systems may  exhibit a variety of long term behaviors. As discussed in e.g. \cite{ZBM98}
    one may have a stable (but not asymptotically stable) dynamics 
    or a mix of stable and asymptotically dynamics. Therefore, the applicability of our method is largely related to which kind of trajectories you wish to obtain. Thus, when we are given a mechanical control system satisfying the conditions of Proposition \ref{orthogonal:input:distribution}, we should first examine the qualitative properties of the associated nonholonomic system. Typical behaviour includes asymptotic stability, periodic or quasi-periodic orbits and conservation of first integrals such as the energy or the nonholonomic momentum. In a wide class of examples, virtual nonholonomic constraints enable us to use energy-momentum methods from \cite{ZBM98} to decide when it is possible to obtain stable or asymptotically stable trajectories. \hfill $\diamond$
\end{remark}

\section{Conclusions}\label{sec:conclusions}
We introduced virtual nonholonomic constraints for mechanical control systems evolving on differentiable manifolds by using an affine connection formalism. We have shown the existence and uniqueness of a control law allowing to define a virtual nonholonomic constraint and we have characterized the trajectories of the closed-loop system as solutions of a mechanical system associated with an induced constrained connection. In addition, we have characterized the dynamics of nonholonomic systems with linear constraints on the velocities in terms of virtual nonholonomic constraints.  In a future work, we would like to extend the results of this paper to nonlinear constraints in order to gain further insigth into the nonlinear nonholonomic virtual constraints defined in \cite{moran2021energy} and \cite{vcelikovsky2021virtual}. In this direction, it would be interesting to impose the energy of the mechanical system as the nonlinear virtual nonholonomic constraint and check if it is possible to design a control keeping the energy constant. Moreover, it would also be interesting to study conditions under which the closed-loop system obtained from Theorem \ref{main:theorem} is equivalent to a nonholonomic system in the same spirit of the approached followed in \cite{ricardo2010control}. Two control systems on a manifold $Q$ of the form
$$\dot{q}=G(q)+u_{a}Y^{a}(q),$$
where $G$ and $Y^{a}$ are vector fields on $Q$, are $S$-equivalent if there exists a diffeomorphism $\phi:Q\rightarrow Q$ such that both their drift vector fields $G$ and control vector fields $Y^{a}$ are $\phi$-related. Then, we may define a control system to be equivalent to a nonholonomic system if it is $S$-equivalent to a mechanical control system for which there exists a control law making its trajectories nonholonomic trajectories. Equivalence is a less restrictive condition than the relation with nonholonomic systems provided in this work. Hence, in principle, it is easier to impose a control law making a control system equivalent to a nonholonomic mechanical system. Though it is a weaker condition, equivalent systems still share the same qualitative behaviour such as stability properties, periodic orbits, etc.


\bibliography{ref}

\appendix
\section{Appendix. Christoffel symbols with constrained connection for Example \ref{disk:example}}\label{appendix} 
The following are the non-vanishing Christoffel symbols:
  \begin{equation*}
        \begin{split}
             \Gamma_{\varphi x}^{x} =& \frac{2Jm\sin\varphi\cos\varphi}{L}-\frac{(IJ+Jm\sin^{2}\varphi)L'}{L^{2}},\\
             \Gamma_{\varphi x}^{y} =& \frac{Jm(\sin^{2}\varphi-\cos^{2}\varphi)}{L}+\frac{J m \sin\varphi \cos\varphi L'}{L^{2}},\\
             \Gamma_{\varphi x}^{\theta} =& \frac{J m \sin\varphi}{L} + \frac{J m\cos\varphi L'}{L^2},\\
            \Gamma_{\varphi x}^{\varphi} =& \frac{m^2 (2 \sin\varphi \cos\varphi + \sin^{2}\varphi - \cos^{2}\varphi)}{L},\\
            & - \frac{m(I + m\sin^{2}\varphi - m \sin \varphi \cos\varphi) L'}{L^2},
                    \end{split}
    \end{equation*}
            \begin{equation*}
            \begin{split}
            \Gamma_{\varphi y}^{x} =& \frac{J m (\sin^{2}{\left(\phi \right)} - \cos^{2}{\left(\phi \right)})}{L} \\&+ \frac{\left(I - J m \sin{\left(\phi \right)} \cos{\left(\phi \right)}\right) L'}{ L^{2}},\\
            \Gamma_{\varphi y}^{y} =& - \frac{2 J m \sin{\left(\phi \right)} \cos{\left(\phi \right)}}{L}+ \frac{\left(- I + J m \cos^{2}{\left(\phi \right)}\right) L'}{L^{2}},\\
            \Gamma_{\varphi y}^{\theta} =& \frac{2 J m^{2} \sin^{2}{\left(\phi \right)} \cos{\left(\phi \right)}}{I L} - \frac{\left(- I m + J m^{2} \cos^{2}{\left(\phi \right)}\right) \cos{\left(\phi \right)}}{I L}\\
            &- \frac{\left(- I m + J m^{2} \cos^{2}{\left(\phi \right)}\right) L' \sin{\left(\phi \right)}}{I L^{2}}\\
            &+ \frac{\left(I m - J m^{2} \sin{\left(\phi \right)} \cos{\left(\phi \right)}\right) \sin{\left(\phi \right)}}{I L}\\
            &- \frac{\left(I m - J m^{2} \sin{\left(\phi \right)} \cos{\left(\phi \right)}\right) L' \cos{\left(\phi \right)}}{I L^{2}}\\
            &- \frac{\left(J m^{2} \sin^{2}{\left(\phi \right)} - J m^{2} \cos^{2}{\left(\phi \right)}\right) \cos{\left(\phi \right)}}{I L},\\
            \Gamma_{\varphi y}^{\varphi} =& \frac{\left(J m^{2} \cos^{2}{\left(\phi \right)} - J m^{2} \sin{\left(\phi \right)} \cos{\left(\phi \right)}\right) L'}{J L^{2}}\\ 
            &+ \frac{m^{2}( \sin^{2}{\left(\phi \right)} - \cos^{2}{\left(\phi \right)} - 2 \sin{\left(\phi \right)} \cos{\left(\phi \right)})}{L},\\
               \Gamma_{\varphi \theta}^{x} = &\frac{I J \sin{\left(\phi \right)} - I \cos{\left(\phi \right)}}{ L} + \frac{\left(- I J \cos{\left(\phi \right)} - I \sin{\left(\phi \right)}\right) L'}{L^{2}},\\
             \Gamma_{\varphi \theta}^{y} =& \frac{I \cos{\left(\phi \right)}}{L} + \frac{I L' \sin{\left(\phi \right)}}{L^{2}},\\
                  \Gamma_{\varphi \theta}^{\theta} =& - \frac{(2+2J) m \sin{\left(\phi \right)} \cos{\left(\phi \right)}}{L} - \frac{m L' \sin^{2}{\left(\phi \right)}}{L^{2}}, \\
            & \quad + \frac{ m (\cos^{2}{\left(\phi \right)} - \sin^{2}{\left(\phi \right)})}{L} \\&+ \frac{\left( J m \cos{\left(\phi \right)} + m \sin{\left(\phi \right)}\right) L' \cos{\left(\phi \right)}}{L^{2}},\\
              \Gamma_{\varphi \theta}^{\varphi} =& \frac{I m \cos{\left(\phi \right)}}{J L} + \frac{I m L' \sin{\left(\phi \right)}}{J L^{2}} + \frac{I J m \sin{\left(\phi \right)} - I m \cos{\left(\phi \right)}}{J L}\\  &+ \frac{\left(- I J m \cos{\left(\phi \right)} - I m \sin{\left(\phi \right)}\right) L'}{J L^{2}},\\
               \Gamma_{\varphi \varphi}^{x} =& - \frac{2 J \sin{\left(\phi \right)} \cos{\left(\phi \right)}}{L} + \frac{\left(- I J - J m \sin^{2}{\left(\phi \right)}\right) L'}{m L^{2}},\\
        \end{split}
    \end{equation*}
            \begin{equation*}
            \begin{split}
            & \Gamma_{\varphi \varphi}^{y} = \frac{J (\cos^{2}{\left(\phi \right)}-\sin^{2}{\left(\phi \right)})}{L} + \frac{J L' \sin{\left(\phi \right)} \cos{\left(\phi \right)}}{L^{2}},\\
            & \Gamma_{\varphi \varphi}^{\theta} = \frac{J m \sin^{3}{\left(\phi \right)}}{I L} - \frac{J m L' \sin^{2}{\left(\phi \right)} \cos{\left(\phi \right)}}{I L^{2}},\\
            & \quad + \frac{\left(- I J - J m \sin^{2}{\left(\phi \right)}\right) \sin{\left(\phi \right)}}{I L} - \frac{\left(- I J - J m \sin^{2}{\left(\phi \right)}\right) L' \cos{\left(\phi \right)}}{I L^{2}},\\
            & \Gamma_{\varphi \varphi}^{\varphi} = - \frac{m \sin^{2}{\left(\phi \right)}}{L} - \frac{2 m \sin{\left(\phi \right)} \cos{\left(\phi \right)}}{L} + \frac{m \cos^{2}{\left(\phi \right)}}{L},\\
            & \quad + \frac{m L' \sin{\left(\phi \right)} \cos{\left(\phi \right)}}{L^{2}} + \frac{\left(- I J - J m \sin^{2}{\left(\phi \right)}\right) L'}{J L^{2}}.
        \end{split}
    \end{equation*}
 
\clearpage 

\end{document}